\documentclass[a4paper,UKenglish,cleveref, autoref, thm-restate]{lipics-v2021}

\hideLIPIcs

\usepackage[noend]{algorithmic}
\usepackage{algorithm}

\newtheorem{wspd}{WSPD-Property}
\crefname{wspd}{WSPD-property}{WSPD-properties}
\Crefname{wspd}{WSPD-Property}{WSPD-Properties}

\DeclareMathOperator*{\odil}{odil}
\newcommand{\bigo}{\ensuremath{\mathcal{O}}}
\newcommand{\oriented}[1]{\overrightarrow{#1}}

\DeclareMathOperator*{\argmin}{arg\,min}

\title{Computing Oriented Spanners and their Dilation}

\author{Kevin Buchin}{Technical University of Dortmund, Germany}{kevin.buchin@tu-dortmund.de}{https://orcid.org/0000-0002-3022-7877}{}
\author{Antonia Kalb}{Technical University of Dortmund, Germany}{antonia.kalb@tu-dortmund.de}{https://orcid.org/0009-0009-0895-8153}{This work was supported by a fellowship of the German Academic Exchange Service (DAAD).}
\author{Anil Maheshwari}{Carleton University, Canada}{anil@scs.carleton.ca}{https://orcid.org/0000-0002-1274-4598}{}
\author{Saeed Odak}{University of Ottawa, Canada}{Saeed.Odak@gmail.com}{}{}
\author{Carolin Rehs}{Technical University of Dortmund, Germany}{carolin.rehs@tu-dortmund.de}{https://orcid.org/0000-0002-8788-1028}{}
\author{Michiel Smid}{Carleton University, Canada}{michiel@scs.carleton.ca}{}{}
\author{Sampson Wong}{BARC, University of Copenhagen, Denmark}{sawo@di.ku.dk}{https://orcid.org/0000-0003-3803-3804}{}

\authorrunning{Buchin, Kalb, Maheshwari, Odak, Rehs, Smid, Wong} 

\Copyright{Kevin Buchin, Antonia Kalb, Anil Maheshwari, Saeed Odak, Carolin Rehs, Michiel Smid, Sampson Wong} 



\category{} 

\relatedversion{} 

\funding{\textit{Anil Maheshwari}, \textit{Michiel Smid}: funded by the Natural Sciences and Engineering Research Council of Canada (NSERC).}

\nolinenumbers 

\begin{document}

\maketitle

\begin{abstract}
Given a point set $P$ in a metric space and a real number $t \geq 1$, an \emph{oriented $t$-spanner} is an oriented graph $\overrightarrow{G}=(P,\overrightarrow{E})$,
where for every pair of distinct points $p$ and $q$ in $P$, the shortest oriented closed walk in $\overrightarrow{G}$ that contains $p$ and $q$ is at most a factor $t$ longer than the perimeter of the smallest triangle in $P$ containing $p$ and $q$. 
The \emph{oriented dilation} of a graph $\overrightarrow{G}$ is the minimum $t$ for which $\overrightarrow{G}$ is an oriented $t$-spanner.

For arbitrary point sets of size $n$ in $\mathbb{R}^d$, where $d \geq 2$ is a constant, the only known oriented spanner construction
is an oriented $2$-spanner with $\binom{n}{2}$ edges. 
Moreover, there exists a set $P$ of four points in the plane, for which the oriented dilation is 
larger than $1.46$, for any oriented graph on $P$. 

We present the first algorithm that computes, in Euclidean space, a sparse oriented spanner whose oriented dilation is bounded by a constant. 
More specifically, for any set of $n$ points in $\mathbb{R}^d$, where $d$ is a constant, we construct an oriented $(2+\varepsilon)$-spanner with $\mathcal{O}(n)$ edges in $\mathcal{O}(n \log n)$ time and $\mathcal{O}(n)$ space. 
Our construction uses the well-separated pair decomposition and an algorithm that computes a $(1+\varepsilon)$-approximation of the minimum-perimeter triangle in $P$ containing two given query points in $\mathcal{O}(\log n)$ time.

While our algorithm is based on first computing a suitable undirected graph and then orienting it, we show that, in general, computing the orientation of an undirected graph that minimises its oriented dilation is NP-hard, even for point sets in the Euclidean plane.

We further prove that even if the orientation is already given, computing the oriented dilation is APSP-hard for points in a general metric space. We complement this result with an algorithm that approximates the oriented dilation of a given graph in subcubic time for point sets in $\mathbb{R}^d$, where $d$ is a constant.
\end{abstract}



\section{Introduction}

While geometric spanners have been researched for decades (see~\cite{Bose.2013,Narasimhan.2007} for a survey), directed versions have only been considered more recently. This is surprising since, in many applications, edges may be directed or even oriented if two-way connections are not permitted.  
Oriented spanners were first proposed in ESA'23~\cite{ESA23} and have since been studied in~\cite{BuchinKLR24.CCCG} and~\cite{BuchinKRS24.EuroCG}.

Formally, let $P$ be a point set in a metric space and let $\oriented{G}$ be an oriented graph with vertex set $P$. For a real number $t \geq 1$, we say that $\oriented{G}$ is an \emph{oriented $t$-spanner}, if for every pair $p,q$ of distinct points in $P$, the shortest oriented closed walk in $\oriented{G}$ that contains $p$ and $q$ is at most a factor $t$ longer than the shortest simple cycle that contains $p$ and $q$ in the complete undirected graph on $P$. Since we consider point sets in a metric space, we observe that the latter is the minimum-perimeter triangle in $P$ containing $p$ and $q$. 

Recall that a $t$-\emph{spanner} is an undirected graph $G$ with vertex set $P$, in which for any two points $p$ and $q$ in $P$, there exists a path between $p$ and $q$ in $G$ whose length is at most $t$ times the distance $|pq|$. Several algorithms are known that compute $(1+\varepsilon)$-spanners with $\bigo(n)$ edges for any set of $n$ points in $\mathbb{R}^d$. Examples are the greedy spanner~\cite{Narasimhan.2007}, $\Theta$-graphs~\cite{KeilG92} and spanners based on the well-separated pair decomposition~(WSPD)~\cite{Smid18}. Moreover, it is NP-hard to compute an undirected $t$-spanner with at most $m$ edges~\cite{Giannopoulos.2010}.

In the oriented case, while it is NP-hard to compute an oriented $t$-spanner with at most $m$ edges \cite{ESA23}, the problem of constructing sparse oriented spanners has remained open until now. 
For arbitrary point sets of size $n$ in $\mathbb{R}^d$, where $d \geq 2$ is a constant, the only known oriented spanner construction is a simple greedy algorithm that computes, in $\bigo(n^3)$ time, an oriented $2$-spanner with $\binom{n}{2}$ edges~\cite{ESA23}. If $P$ is the set in $\mathbb{R}^2$, consisting of the three vertices of an equilateral triangle and a fourth point in its centre, then the smallest oriented dilation for $P$ is $2 \sqrt{3} - 2 \approx 1.46$~\cite{ESA23}. 
Thus, prior to our work, no algorithms were known that compute an oriented $t$-spanner with a subquadratic number of edges for any constant $t$. 

In this paper, we introduce the first algorithm to construct sparse oriented spanners for general point sets. For any set $P$ of $n$ points in $\mathbb{R}^d$, where $d$ is a constant, the algorithm computes an oriented $(2+\varepsilon)$-spanner with $\mathcal{O}(n)$ edges in $\bigo(n \log n)$ time and $\bigo(n)$ space.

While our approach uses the WSPD~\cite{CallahanKosaraju95WSPD, Smid18}, this does not immediately yield an oriented spanner. An undirected spanner can be obtained from a WSPD by adding one edge per
well-separated pair. This does not work for constructing an oriented spanner because a path in both directions needs to be considered for every well-separated pair.
To construct such paths, we present a data structure that, given a pair of points, returns a point with an approximate minimum summed distance to the pair of the points, i.e.\@ the smallest triangle in the complete undirected graph on $P$. Our algorithm first constructs an adequate undirected sparse graph and subsequently orients it using a greedy approach. 

We also show in this paper that, in general, it is NP-hard to decide whether the edges of an arbitrary graph can be oriented in such a way that the oriented dilation of the resulting graph is below a given threshold. This holds even for Euclidean graphs.

For the problem of computing the oriented dilation of a given graph, there is a straightforward cubic time algorithm. We show a subcubic approximation algorithm.

The hardness of computing the dilation, especially in the undirected case, is a long-standing open question, and only cubic time algorithms are known~\cite{gudmundsson2018.dilation}. It is conjectured that computing the dilation is APSP-hard. In this paper, we demonstrate that computing the oriented dilation is APSP-hard.

\section{Preliminaries}

Let $(P,|\cdot|)$ be a finite metric space, where the distance between any two points $p$ and $q$ in $P$ is denoted by $|pq|$. 
 When the choice between Euclidean and general metric distances is relevant for our results, we use the adjective \emph{Euclidean} or respectively \emph{metric graph}.

 If $p$ and $q$ are distinct points in $P$, then $\Delta^*(p,q)$ denotes the shortest simple cycle containing $p$ and $q$ in the complete undirected graph on $P$; this is the triangle $\Delta_{pqx}$, where 
\[ x = \arg\min \{ \left( |px| + |qx| \right) \mid  
               x\in P \setminus \{p,q\}\} .
\]
We use $| \Delta^*(p,q) |$ to denote the perimeter of this triangle. 

An \emph{oriented graph} $\oriented{G}=(P,\oriented{E})$ is a directed graph such that for any two distinct points $p$ and $q$ in $P$, at most one of the two directed edges $(p,q)$ and $(q,p)$ is in $\oriented{E}$. In other words, if $(p,q)$ is in $\oriented{E}$, then $(q,p)$ is not in $\oriented{E}$.

For two distinct points $p$ and $q$ in $P$, we denote by $C_{\oriented{G}}(p,q)$ the \emph{shortest closed walk} containing $p$ and $q$ in the oriented graph $\oriented{G}$. The length of this walk is denoted by $| C_{\oriented{G}}(p,q) |$; if such a walk does not exist, then $| C_{\oriented{G}}(p,q) | = \infty$.

\begin{definition}[oriented dilation]\label{def:odilation}
Let $(P,|\cdot|)$ be a finite metric space and let $\oriented{G}$ be an oriented graph with vertex set $P$. 
For two points $p,q\in P$, their \emph{oriented dilation} is defined as
   \begin{equation*}
       {\odil}_{\oriented{G}} (p,q)=\frac{|C_{\oriented{G}}(p,q)|}{|\Delta^*(p,q)|}.
   \end{equation*}
    The oriented dilation of $\oriented{G}$ is defined as 
$\odil(\oriented{G})={\max} \{{\odil}_{\oriented{G}}(p,q)\mid p,q\in P\}$.
\end{definition}

An oriented graph with oriented dilation at most $t$ is called an \emph{oriented $t$-spanner}. 
Note that
there is no oriented $1$-spanner for sets of at least $4$ points.

 Let $G=(P,E)$ be an undirected graph. We call $\oriented{G}=(P,\oriented{E})$ an \emph{orientation} of $G$, if for each undirected edge $\{p,q\}\in E$ it holds either $(p,q)\in \oriented{E}$ or $(q,p)\in \oriented{E}$.


\subsection{Well-Separated Pair Decomposition}

In 1995, Callahan and Kosaraju~\cite{CallahanKosaraju95WSPD} introduced the \emph{well-separated pair decomposition}~(WSPD) as a tool to solve various distance problems on multidimensional point sets. 

For a real number $s>0$, two point sets $A$ and $B$ form an \emph{$s$-well-separated pair}, if $A$ and $B$ can each be covered by balls of the same radius, say $\rho$, such 
that the distance between the balls is at least $s \cdot \rho$. 
An $s$-\emph{well-separated pair decomposition} for a set $P$ of points is a sequence of $s$-well-separated pairs $\{A_i,B_i\}$, such that $A_i \cap B_i = \emptyset$
and for any two points $p,q \in P$, there is a unique index $i$ such that $p \in A_i$ and $q \in B_i$, or $p \in B_i$ and $q \in A_i$.

The following properties are used repeatedly in this paper.

\begin{wspd}[Smid~\cite{Smid18}]\label{theo:wspd-properties}
    Let $s>0$ be a real number and let $\{A,B\}$ be an $s$-well-separated pair. For points $a,a'\in A$ and $b,b'\in B$, it holds that
    \begin{enumerate}[(i)]
        \item $|aa'|\leq 2/s\cdot |ab|$, and 
        \item $|a'b'|\leq (1+4/s)|ab|$.
    \end{enumerate}
\end{wspd}

\noindent Callahan and Kosaraju~\cite{CallahanKosaraju95WSPD} showed a fast WSPD-construction based on a split tree. 

\begin{wspd}[Callahan and Kosaraju~\cite{CallahanKosaraju95WSPD}] \label{theo:wspd-construction}
   Let $P$ be a set of $n$ points in $\mathbb{R}^d$, where $d$ is a constant, and let $s>0$ be a real number. An $s$-well-separated pair decomposition for $P$ consisting of $\bigo(s^d n)$ pairs, can be computed in $\bigo(n\log n + s^d n)$ time.
\end{wspd}


\subsection{Approximate Nearest Neighbour Queries}\label{sec:ANN}

For approximate nearest neighbour queries, we want to preprocess a given point set $P$ in a metric space, such that for any query point $q$ in the metric space, we can report quickly its \emph{approximate nearest neighbour} in $P$. That is, if $q^*$ is the exact nearest neighbour of $q$ in $P$, then the query algorithm returns a point $q'$ in $P$ such that 
$|qq'| \leq (1+\varepsilon) \cdot |qq^*|$. 
Such a data structure is offered in~\cite{AryaMountDS1998} for the case when $P$ is a point set in the Euclidean space $\mathbb{R}^d$: 

\begin{lemma}[Arya, Mount, Netanyahu, Silverman and Wu~\cite{AryaMountDS1998}]\label{theo:nearest-neighbour-query}
    Let $P$ be a set of $n$ points in $\mathbb{R}^d$, where $d$ is a constant. In $\bigo(n \log n)$ time, a data structure of size 
    $\bigo(n)$ can be constructed that supports the following operations:
    \begin{itemize}
        \item Given a real number $\varepsilon >0$ and a query point $q$ in $\mathbb{R}^d$, return a $(1+\varepsilon)$-approximate nearest neighbour of $q$ in $P$ in $\bigo(\log n)$ time.
        \item Inserting a point into $P$ and deleting a point from $P$ in $\bigo(\log n)$ time.
    \end{itemize}
\end{lemma}

\section{Sparse Oriented Constant Dilation Spanner}
\label{sec:2+eps-spanner}



This section presents the first construction for a sparse oriented $(2+\varepsilon)$-spanner for general point sets in $\mathbb{R}^d$. 
Our algorithm consists of the following four steps:
\begin{enumerate}
    \item Compute the WSPD on the given point set.
    \item For two points of each well-separated pair, approximate their minimum-perimeter triangle by our algorithm using approximate  nearest neighbour queries presented in \cref{sec:min-triangle}. 
    \item Construct an undirected graph whose edge set consists of edges of approximated triangles. 
    \item Orient the undirected graph using a greedy algorithm (see \cref{theo:2+eps-orientation}).
\end{enumerate}

In Step~4, we modify the following greedy algorithm from~\cite{ESA23}. 
Sort the $\binom{n}{3}$ triangles of the complete graph in ascending order by their perimeter. 
For each triangle in this order, orient it clockwise or anti-clockwise if possible; otherwise, we skip this triangle.
At the end of the algorithm, all remaining edges are oriented arbitrarily. 
 In~\cite{ESA23}, they show that this greedy algorithm yields an oriented $2$-spanner. Note that this graph has $\binom{n}{2}$ edges. 

 
Our modification of the greedy algorithm is to only consider the approximate minimum-perimeter triangles, rather than all $\binom{n}{3}$ triangles (see \cref{alg:2+eps-orientation}). \Cref{theo:2+eps-orientation} states that, for any two points whose triangles are $(1+\varepsilon)$-approximated in Step~2, the dilation between those points is at most~$(2+2 \varepsilon)$ in the resulting orientation.

\begin{lemma}\label{theo:2+eps-orientation}
     Let $P$ be a set of $n$ points in $\mathbb{R}^d$, let $\varepsilon_1>0$ be a real number, and let $L$ be a list of $m$ point triples $(p,q;r)$, with $p$, $q$, and $r$ being pairwise distinct points in $P$. Assume that for each $(p,q;r)$ in $L$, 
     \[ |\Delta_{pqr}| \leq (1+\varepsilon_1)\cdot|\Delta^*(p,q)| . 
     \]
     Let $G=(P,E)$ be the undirected graph whose edge set consists of all edges of the $m$ triangles $\Delta_{pqr}$ defined by the triples $(p,q;r)$ in $L$. 
     In $\bigo(m \log n)$ time, we can compute an orientation $\oriented{G}$ of 
     $G$, such that $\odil_{\oriented{G}}(p,q)\leq 2+2\varepsilon_1$ for each $(p,q;r)$ in $L$.
\end{lemma}
\begin{proof}
 We prove that \cref{alg:2+eps-orientation} computes such an orientation $\oriented{G}$. 
Let $(p,q;r)\in L$ be a point triple. By definition, 
 $|\Delta_{pqr}|\leq (1+\varepsilon_1)\cdot|\Delta^*(p,q)|$.
    
   If $\Delta_{pqr}$ is oriented consistently then $\odil_{\oriented{G}}(p,q)=1+\varepsilon_1$. 
   Otherwise, at least two of the edges of $\Delta_{pqr}$ were already oriented, when \cref{alg:2+eps-orientation} processed $\Delta_{pqr}$.  Therefore those edges are incident to previously processed, consistently oriented triangles. 
   By concatenating those triangles, we obtain a closed walk containing $p$ and $q$. Since both triangles are processed first, their length is at most $|\Delta_{pqr}|$ and the closed walk is at most twice as long. Thus, $\odil_{\oriented{G}}(p,q) \leq 2\cdot(1+\varepsilon_1)= 2+2\varepsilon_1$

   Sorting the triangles costs $\bigo(m\log m)$ time. We store the points of $P = \{p_1,p_2,\ldots,p_n\}$ in an array. With each entry $i$, we store all edges $(p_i,p_j)$ that are incident to $i$. We store these in a balanced binary search tree, sorted by $j$.  By this, an edge can be accessed and oriented in $\bigo(\log n)$ time. Since $m=\bigo(n)$, the runtime is $\bigo(m (\log m + \log n))=\bigo(m \log n)$.
\end{proof}

\begin{algorithm}[ht]
	\caption{Greedy $(2+2\varepsilon_1)$-Orientation}\label{alg:2+eps-orientation}
     \begin{algorithmic}
    		\REQUIRE set $P$ of $n$ points, real $\varepsilon_1>0$,\\
    		list $L$ of point triples $(p,q;r)$, s.t.\ for each $(p,q;r)\in L$ holds $|\Delta_{pqr}|\leq (1+\varepsilon_1)\cdot|\Delta^*(p,q)|$
    		\ENSURE oriented graph  $G=(P,E)$ whose edge set contains all edges of the triangles $\Delta_{pqr}$ where $(p,q;r)\in L$ and for each $(p,q;r)\in L$ holds $\odil_G(p,q)\leq 2+2\varepsilon_1$

        \STATE Sort the triangles $\Delta_{pqr}$ of the triples in $L$  ascending by lengths 
        \FORALL{$\Delta_{pqr}$ with $(p,q;r)\in L$}
            \IF{no edge of $\Delta_{pqr}$ is oriented}
              \STATE orient $(p,q)$ arbitrarily and orient $\Delta_{pqr}$ consistently according to $(p,q)$
           \ELSE
              \STATE orient $\Delta_{pqr}$ consistently, if this is still possible
            \ENDIF
        \ENDFOR
        \STATE orient the remaining edges arbitrarily
\end{algorithmic}
\end{algorithm}

\begin{remark*}
    The above lemma only guarantees an upper bound on 
    $\odil_{\oriented{G}}(p,q)$ for $(p,q;r)$ in~$L$. For such a triple, 
    $\odil_{\oriented{G}}(p,r)$ and 
    $\odil_{\oriented{G}}(q,r)$ can be arbitrarily large. 
\end{remark*}







\begin{algorithm}[ht]
	\caption{Sparse $(2+\varepsilon)$-Spanner}
    \label{alg:wspd-oriented}
     \begin{algorithmic}
    		\REQUIRE Set $P$ of $n$ points in $\mathbb{R}^d$, positive reals $\varepsilon<2$, $\varepsilon_1$, $s$
    		\ENSURE $(2+\varepsilon)$-spanner for $P$ with $\bigo(n)$ edges

        \STATE Compute an $s$-WSPD $\{A_i,B_i\}$ for $1 \leq i \leq m = \bigo(n)$
        \STATE Compute the data structure $M$ on $P$ to approximate nearest neighbours (see \cref{theo:nearest-neighbour-query}~\cite{AryaMountDS1998})%
        \STATE $L=\emptyset$

        \FOR{$i=1$ \textbf{to} $m$}
            \STATE Pick $\min\{|A_i|,2\}$ points of $A_i$
            \STATE Pick $\min\{|B_i|,2\}$ points of $B_i$
            \STATE $K=$ set of all picked points
            \FORALL{$\{p,q\}\in K\times K$ with $p\neq q$}
                \STATE $\Delta_{pqr}=$ $(1+\varepsilon_1)$-approximation of $\Delta^*(p,q)$ computed by \cref{alg:approx-min-triangle} given $M,p,q, \varepsilon_1$%
                \STATE Add the point triple $(p,q;r)$ to $L$
            \ENDFOR
        \ENDFOR
        
        \RETURN  oriented graph generated by \cref{alg:2+eps-orientation}, 
        given $P, \varepsilon_1, L$
\end{algorithmic}
\end{algorithm}

Now, we show that \cref{alg:wspd-oriented} constructs a sparse oriented $(2+\varepsilon)$-spanner for point sets in the Euclidean space $\mathbb{R}^d$:




\begin{theorem}\label{theo:2+eps-spanner}
    Given a set $P$ of $n$ points in $\mathbb{R}^d$ and a real number $0<\varepsilon<2$, a $(2+\varepsilon)$-spanner for~$P$ with $\bigo(n)$ edges can be constructed in $\bigo(n \log n)$ time and $\bigo(n)$ space.
\end{theorem}
\begin{proof}
  Let the constants in \cref{alg:wspd-oriented} be equal to $\varepsilon_1={\varepsilon}/{4}$ and $s={96}/{\varepsilon}$.
  Let $\oriented{G}$ be the graph returned by \cref{alg:wspd-oriented}. 
   
    For any two points $p,q\in P$, we prove that $|C_{\oriented{G}}(p,q)|\leq (2+\epsilon)\cdot|\Delta^*(p,q)|$. The proof is by induction on the rank of $|pq|$ in the sorted sequence of all $\binom{n}{2}$ pairwise distances.

    If $p,q$ is a closest pair, then the pair $\{\{p\},\{q\}\}$ is in the WSPD. Therefore, a $(1+\varepsilon_1)$-approximation of $\Delta^*(p,q)$ is added to the list $L$ of triangles, that define the edge set of $\oriented{G}$,  and, due to \cref{theo:2+eps-orientation} and $\varepsilon_1={\varepsilon}/{4}$, we have
     $  \odil_{\oriented{G}} (p,q) \leq 2+2\varepsilon_1 \leq 2+\varepsilon.$

Assume that $p,q$ is not a closest pair. 
    Let $\{A,B\}$ be the well-separated pair with $p\in A$ and $q\in B$. We distinguish cases based on whether $p$ and/or $q$ are picked by \cref{alg:wspd-oriented} while processing the pair $\{A,B\}$.

\subparagraph{Case 1:} Neither $p$ nor $q$ are picked by \cref{alg:wspd-oriented}.

This implies $|A|\geq 3$ and $|B|\geq 3$. Let $a\in A$ and $b\in B$ be picked points. We prove 
        \[|C_{\oriented{G}}(p,q)|\leq |C_{\oriented{G}}(p,a)|+|C_{\oriented{G}}(a,b)|+|C_{\oriented{G}}(b,q)|\leq (2+\varepsilon)\cdot|\Delta^*(p,q)|.\]
    A sketch of this proof is visualized in \cref{fig:2+eps-spanner-proof-sketch}.

\begin{figure}[ht]
    \centering
    \includegraphics{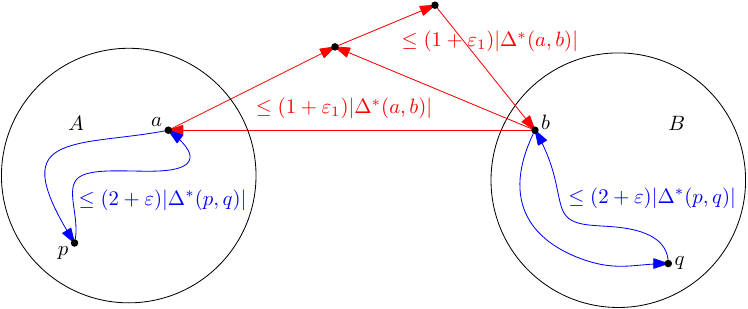}
    \caption{Visualization of Case 1: $|C_{\oriented{G}}(p,q)|\leq |C_{\oriented{G}}(p,a)|+|C_{\oriented{G}}(a,b)|+|C_{\oriented{G}}(b,q)|$ in the proof of \cref{theo:2+eps-spanner}. $C_{\oriented{G}}(a,b)$ is either a $(1+\varepsilon_1)$-approximation of $\Delta^*(a,b)$ or a union of two triangles (red), each of at most $(1+\varepsilon_1)|\Delta^*(a,b)|$. $|C_{\oriented{G}}(p,a)|$ and $|C_{\oriented{G}}(b,q)|$ (blue) are bounded inductively.}
    \label{fig:2+eps-spanner-proof-sketch}
\end{figure}
    
     Since $|pa|<|pq|$ and $|bq|<|pq|$, we apply induction to bound $|C_{\oriented{G}}(p,a)|$ and $|C_{\oriented{G}}(b,q)|$.
     Since $a$ and $b$ are picked points, \cref{theo:2+eps-orientation} applies, which gives us
     \begin{equation*}
         |C_{\oriented{G}}(p,q)|\leq (2+\varepsilon) |\Delta^*(p,a)| +(2+2\varepsilon_1) |\Delta^*(a,b)| +  (2+\varepsilon) |\Delta^*(b,q)|.
     \end{equation*}       

Since $|A|\geq 3$, $|\Delta^*(p,a)|$ is at most the perimeter of any triangle with vertices $p$, $a$, and one other point in $A$.  Each of the three edges of such a triangle can be bounded using \cref{theo:wspd-properties}. It holds that        
\begin{equation}\label{eq:bound-Delta-same-set-pair}
  |\Delta^*(p,a)| \leq 3\cdot 2/s \cdot  |pq| \leq  3/s \cdot |\Delta^*(p,q)|,
\end{equation}
where the last inequality $|\Delta^*(p,q)|\geq 2 |pq|$ follows from the triangle inequality.
Analogously, we bound $|\Delta^*(b,q)|$ by $3/s|\Delta^*(p,q)|$. 
By \cref{theo:delta-ab-WSPD}, we bound $|\Delta^*(a,b)|$ and achieve in total
    \begin{equation*}
         |C_{\oriented{G}}(p,q)|\leq 2 (2+\varepsilon)\cdot 3/s |\Delta^*(p,q)| +(2+2\varepsilon_1)\left(1+8/s\right) |\Delta^*(p,q)|.
     \end{equation*}

    Since $\varepsilon_1=\varepsilon/4$, $\varepsilon < 2$ and  $s= 96/\varepsilon$, we conclude that 
\begin{equation*}
     |C_{\oriented{G}}(p,q)|\leq \left(2+\frac{\varepsilon}{2}+\frac{28+10\varepsilon}{s}\right)|\Delta^*(p,q)| \leq \left(2+\frac{\varepsilon}{2}+\frac{48}{s}\right)|\Delta^*(p,q)| {=} (2+\varepsilon)|\Delta^*(p,q)|.
 \end{equation*}

\subparagraph{Case 2:} Either $p$ or $q$ is picked by \cref{alg:wspd-oriented}.

W.l.o.g., let $p$ and a point $b\in B$ be picked. Analogously to Case 1, it holds that 
\begin{align*}
    |C_{\oriented{G}}(p,q)|&\leq |C_{\oriented{G}}(p,b)|+|C_{\oriented{G}}(b,q)|
    \leq  (2+\varepsilon)|\Delta^*(p,q)|.
\end{align*}
   
\subparagraph{Case 3:} Both $p$ and $q$ are picked by \cref{alg:wspd-oriented}.

An approximation of $\Delta^*(p,q)$ is added to $L$ and then, due to \cref{theo:2+eps-orientation} and $\varepsilon_1=\frac{\varepsilon}{4}$, we have $\odil_{\oriented{G}}(p,q)\leq  2+2\varepsilon_1 \leq 2+\varepsilon$.\bigskip
 
     The initialisation is dominated by computing a WSPD in $\bigo(n \log n)$ time~\cite{CallahanKosaraju95WSPD}. 
     For each of the $\bigo(n)$ pairs, we pick at most four points and approximate the minimum triangle $\Delta^*(p,q)$ for at most six pairs $\{p,q\}$ of picked points. Therefore, computing the list $L$ of $\bigo(n)$ triples costs $\bigo(n \log n)$ time (see \cref{theo:approx-triangle-queries}) and orienting those costs $\bigo(n \log n)$ time (\cref{theo:2+eps-orientation}). Therefore, a $(2+\varepsilon)$-spanner can be computed in $\bigo(n \log n)$ time and $\bigo(n)$ space.
\end{proof}

\subsection{The Minimum-Perimeter Triangle}\label{sec:min-triangle}

Let $P$ be a set of $n$ points in $\mathbb{R}^d$, and let $p$ and $q$ be two distinct points in $P$. The minimum triangle $\Delta^*(p,q)$ containing $p$ and $q$ can be computed naively in $O(n)$ time. 

The following approach can be used to answer such a query in $O(n^{1-c})$ time, for some small constant $c>0$ that depends on the dimension $d$. Note that computing $\Delta^*(p,q)$ is equivalent to computing the smallest real number $\rho>0$, such that the ellipsoid
\[ \sqrt{\sum_{i=1}^d (p_i-x_i)^2} + \sqrt{\sum_{i=1}^d (q_i-x_i)^2} 
     \leq \rho  
\]
contains at least three points of $P$. 

For a fixed real number $\rho$, we can count the number of points of $P$ in the above ellipsoid: Using \emph{linearization}, see Agarwal and Matousek~\cite{AgarwalM94}, we convert $P$ to a point set $P'$ in a Euclidean space whose dimension depends on $d$. The problem then becomes that of counting the number of points of $P'$ in a query simplex. Chan~\cite{Chan12} has shown that such queries can be answered in $O(n^{1-c})$ time, for some small constant $c>0$.
This result, combined with \emph{parametric search}, see Megiddo~\cite{Megiddo83}, allows us to compute $\Delta^*(p,q)$  in $O(n^{1-c})$ time.\bigskip

For many applications, including our $(2+\varepsilon)$-spanner construction, an approximation of the minimum triangle is sufficient. The following lemma shows that a WSPD with respect to $s = 2/\varepsilon$ provides a $(1+\varepsilon)$-approximation. 

\begin{lemma}
\label{theo:delta-ab-WSPD}
 Let $A$ and $B$ be subsets of  a point set $P$ such that $\{A,B\}$ is an $s$-well-separated pair, and assume that $|B|\geq 2$. 
 \begin{enumerate}
     \item For any point $a\in A$ and any two points $b,c\in B$, we have $|\Delta_{abc}| \leq (1+2/s)  |\Delta^*(a,b)|$. 
     \item  Assume that $|A|\geq 2$. For any two points $a,a'\in A$ and any two points $b,b'\in B$, it holds $|\Delta^*(a,b)| \leq (1+8/s) |\Delta^*(a',b')|$.
  \end{enumerate}
\end{lemma}
\begin{proof}
 By the triangle inequality, we have
\[ | \Delta^*(a,b) | \geq 2|ab| . 
\]
Again using the triangle inequality, we have

\[
 | \Delta_{abc} |  =  |ab| + |bc| + |ac| 
   \leq  |ab| + |bc| + |ab| + |bc| 
   =  2|ab| + 2|bc| . 
\]

Using \cref{theo:wspd-properties}, we conclude that
\[ | \Delta_{abc} | \leq (2+4/s) |ab| \leq 
   (1 + 2/s) | \Delta^*(a,b) | .
\]

Since the optimal triangle containing two points is always shorter than another triangle containing those points, we apply twice the previous statement to conclude that
\[
 |\Delta^*(a,b)|\leq \left(1+\frac{2}{s}\right) \cdot |\Delta^*(a',b)| \leq  \left(1+\frac{2}{s}\right)^2 \cdot |\Delta^*(a',b')| \leq \left(1+\frac{8}{s}\right) \cdot |\Delta^*(a',b')|. \qedhere
\]
\end{proof}

 Note that this works only if at least one subset of the well-separated pair contains at least two points. If both subsets have size one, we can use the following lemma to approximate the minimum triangle using approximate nearest neighbour queries.

\begin{lemma}\label{theo:approx-triangle-queries}
    Let $P$ be a set of $n$ points in $\mathbb{R}^d$ and let $0<\varepsilon_1<2$ be a real number. 
In $\bigo(n\log n)$ time, the set $P$ can be preprocessed into a data structure of size $\bigo(n)$, such that, given any two distinct query points $p$ and $q$ in $P$, a $(1+\varepsilon_1)$-approximation of the minimum triangle $\Delta^*(p,q)$ can be computed in $\bigo(\log n)$ time.
\end{lemma}

\begin{proof}
The data structure is the approximate nearest neighbour structure of 
\cref{theo:nearest-neighbour-query}~\cite{AryaMountDS1998}. The query algorithm is presented in 
\cref{alg:approx-min-triangle}. 
Let the constants in this algorithm be equal to 
$\varepsilon_2 = \varepsilon_1/2$, $\alpha = 4/\varepsilon_1$, and 
$\varepsilon_3=\frac{2}{3 \sqrt{d}}\cdot \varepsilon_1$. 

Let $p$ and $q$ be two distinct points in $P$. 
We prove that \cref{alg:approx-min-triangle} returns a $(1+\varepsilon_1)$-approximation of $\Delta^*(p,q)$.

\begin{algorithm}[ht]
	\caption{$(1+\varepsilon_1)$-Approximation of the Minimum-Perimeter Triangle $\Delta^*(p,q)$}\label{alg:approx-min-triangle}
    \begin{algorithmic}
    		\REQUIRE Query points $p,q\in P$, positive reals $\varepsilon_1<2$, $\varepsilon_2$, $\varepsilon_3$, $\alpha>2$, data structure $M$ to approximate nearest neighbours
    		\ENSURE Triangle $\Delta$ such that $|\Delta|\leq (1+\varepsilon_1)|\Delta^*(p,q)|$
     \STATE Delete $p$ and $q$ from $M$
    \STATE $r=$ $(1+\varepsilon_2)$-approximate nearest neighbour of $p$ in $M$
    \IF{$|pr|>\alpha|pq|$}
    \RETURN $\Delta_{pqr}$
    \ELSE 
    \STATE Let $B$ be the hypercube with sides of length $3\alpha|pq|$ that is centred at $p$
        \STATE Divide $B$ into cells with sides of length of $\varepsilon_3|pq|$
        \FORALL{cell $C$ in $B$}
        \STATE $x_c=$ $(1+\varepsilon_2)$-approximate nearest neighbour of the centre $c$ of $C$
        \ENDFOR
        \RETURN $\Delta_{pqx_c}$ with $x_c=\underset{C\in B}{\argmin} |\Delta_{pqx_c}|$
    \ENDIF
    \STATE Re-Insert $p$ and $q$ in $M$
    \end{algorithmic}
\end{algorithm}

Throughout this proof, we denote the third point of $\Delta^*(p,q)$ by $x$. Let $r$ be the $(1+\varepsilon_2)$-approximate nearest neighbour of $p$ in $P \setminus \{p,q\}$ that is computed by the algorithm. 
By the triangle inequality, it holds

\begin{equation} 
\label{eq:approx-triangle-triangle-inequality} 
 |\Delta_{pqr}| = |pq|+|qr|+|pr| \leq 2|pq|+2|pr| . 
\end{equation} 

\subparagraph{Case 1:} $|pr| > \alpha |pq|$. 

Let $p^*$ be the exact nearest neighbour of $p$ in the set $P \setminus \{p,q\}$. Since $r$ is a $(1+\varepsilon_2)$-approximate nearest neighbour of $p$ in $P \setminus \{p,q\}$, and since $|\Delta^*(p,q)| \geq 2|px|$, we have 
\[ |pr| \leq (1+\varepsilon_2) |pp^*| \leq (1+\varepsilon_2) |px|  
    \leq (1+\varepsilon_2) |\Delta^*(p,q)|/2 .
\] 
Combining this inequality with \cref{eq:approx-triangle-triangle-inequality} and the assumption that 
$|pr| > \alpha |pq|$, we have 
\[ |\Delta_{pqr}| \leq (2+2/\alpha) |pr| \leq 
    (1+1/\alpha) (1+\varepsilon_2) |\Delta^*(p,q)| . 
\]
Since $\varepsilon_2 = \varepsilon_1/2$, $\alpha = 4/\varepsilon_1$ and $\varepsilon_1 < 2$, 
we have 
\[ |\Delta_{pqr}| \leq (1+\varepsilon_1/4) (1 + \varepsilon_1/2) 
            | \Delta^*(p,q)| = 
  ( 1+ 3 \varepsilon_1/4 + \varepsilon_1^2 / 8 ) | \Delta^*(p,q)|
   \leq (1 + \varepsilon_1)  | \Delta^*(p,q)|.
\] 

\begin{figure}[ht]
    \centering
    \includegraphics[page=1]{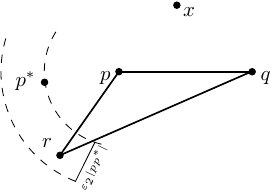}
    \caption{Visualization of Case 1 in the proof of \cref{theo:approx-triangle-queries}. Since $r$ is a $(1+\varepsilon_2)$-approximated nearest neighbour of $p$, the returned triangle $\Delta_{pqr}$ is a $(1+\varepsilon_1)$-approximation of $\Delta^*(p,q)=\Delta_{pqx}$.}
    \label{fig:1+eps-approx-triangle-case1}
\end{figure}

\subparagraph{Case 2:} $|pr| \leq \alpha |pq|$. 

Consider the hypercube $B$ with sides of length $3 \alpha |pq|$ 
that is centred at $p$. 

We first prove that the third point $x$ of $\Delta^*(p,q)$ is in $B$. 
Using \cref{eq:approx-triangle-triangle-inequality}, the assumption that $|pr| \leq \alpha |pq|$, 
and the fact that $\alpha \geq 2$, we have 
\[ |\Delta^*(p,q)| \leq |\Delta_{pqr}| \leq (2+2\alpha) |pq|
            \leq 3 \alpha |pq| .
\] 
Since $|\Delta^*(p,q)| \geq 2|px|$, it follows that 
$|px| \leq (3\alpha/2) |pq|$ and, therefore, $x \in B$. 

Let $C$ be the cell of $B$ that contains $x$, let $c$ be the centre of $C$, let $c^*$ be the exact nearest neighbour of $c$ in $P \setminus \{p,q\}$, and let $x_c$ be the $(1+\varepsilon_2)$-approximate nearest neighbour of $c$ that is computed 
by the algorithm. 
We first prove an upper bound on the distance $|x x_c|$: 
\[ |x x_c| \leq |xc| + |c x_c| \leq |xc| + (1+\varepsilon_2) |cc^*|  
  \leq |xc| + (1+\varepsilon_2) |cx| = (2+\varepsilon_2) |cx|  
    \leq 3 |cx| . 
\]
Since $x$  and $c$ are in the same $\varepsilon_3 |pq| \times \varepsilon_3 |pq|$-sized cell and $\varepsilon_3=\frac{2}{3 \sqrt{d}}\cdot \varepsilon_1$, it follows that 
\[ |x x_c| \leq 3 \cdot\sqrt{d}/2  \varepsilon_3 |pq| = 
       \varepsilon_1 |pq| \leq (\varepsilon_1/2) |\Delta^*(p,q)|.
\]
where the last inequality $|\Delta^*(p,q)|\geq 2 |pq|$ follows from the triangle inequality.

Let $\Delta_{pqx_{c'}}$ be the triangle that is returned by the 
algorithm. Then 
\begin{equation*}
     |\Delta_{pqx_{c'}}| \leq |\Delta_{pqx_c}|  \leq  |pq| + |qx| + |x x_c| + |x_c x| + |xp|  = | \Delta^*(p,q)| + 2 |x x_c |  \leq  (1+\varepsilon_1) | \Delta^*(p,q)| . 
\end{equation*}

\begin{figure}[ht]
    \centering
    \includegraphics[page=2]{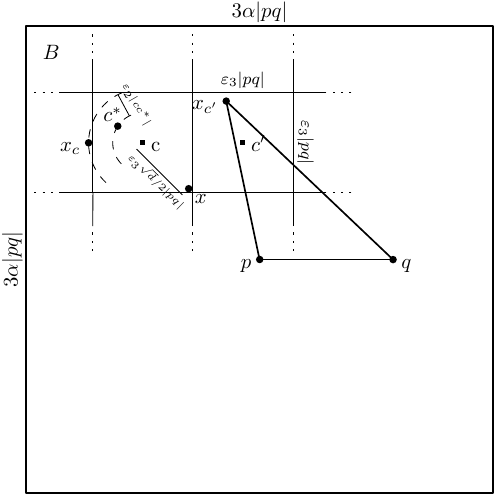}
    \caption{Visualization of Case 2 in the proof of \cref{theo:approx-triangle-queries}. The third point $x$ of $\Delta^*(p,q)$ lies in the cell centred at $c$. Since $\Delta_{pq{x_{c'}}}$ is returned, it holds  $|\Delta_{pq{x_{c'}}}|\leq |\Delta_{pq{x_c}}|$, where $x_{c'}$ and $x_c$ are $(1+\varepsilon_2)$-approximated nearest neighbours of $c'$ and $c$.}
    \label{fig:1+eps-approx-triangle-case2}
\end{figure}

The query algorithm performs two deletions, two insertions, and $1 + (3 \alpha / \varepsilon_3)^d = \bigo(1)$ approximate nearest neighbour queries in the data structure of \cref{theo:nearest-neighbour-query}~\cite{AryaMountDS1998}. Each such operation takes $\bigo(\log n)$ time. Therefore, the entire query algorithm takes $\bigo(\log n)$ time.
\end{proof}

\subsection{NP-Hardness}

Our $(2+\varepsilon)$-spanner construction involves computing undirected $(1+\varepsilon)$-approximations of minimum-perimeter triangles, which are subsequently oriented. A similar approach is used in~\cite{ESA23} where they prove, for convex point sets, orienting a greedy triangulation yields a plane oriented $O(1)$-spanner.
Both results could be improved by optimally orienting the underlying undirected graph. 
The authors in~\cite{ESA23} pose, as an open question, whether it is possible to compute, in polynomial time, an orientation of any given undirected geometric graph, that minimises the oriented dilation. 
We answer this question by showing NP-hardness.


\begin{theorem}\label{theo:np-hard-decision-orient}
Let $P$ be a finite set of points in a metric space and let $G$ be an undirected graph with vertex~$P$.
Given a real number $t'$, it is NP-hard to decide if there exists an orientation $\oriented{G}$ of $G$ with oriented dilation $\odil(\oriented{G})\leq t'$. 
This is even true for point sets in the Euclidean plane. 
\end{theorem}
\begin{proof}
   We reduce from the NP-complete problem planar $3$-SAT~\cite{Lichtenstein82}. We start with a planar Boolean formula $\varphi$ in conjunctive normal form with an incidence graph $G_\varphi$ as illustrated in \cref{fig:example-planar-3SAT}.
   This graph can be embedded on a polynomial-size grid whose cells have side length one~\cite{Godau1999,Lichtenstein82}. 
     We give a construction for a graph $G$ based on $G_\varphi$ such that there is an orientation $\overrightarrow{G}$ of $G$ with dilation $\odil(\overrightarrow{G})\leq t'$ with $t' := 1.043$ if and only if $\varphi$ is satisfiable. 

\begin{figure}[ht]
    \centering
    \includegraphics{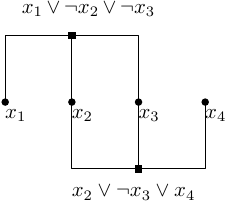}
    \caption{Example: Incidence graph of a planar $3$-SAT formula embedded on a square grid}
    \label{fig:example-planar-3SAT}
\end{figure}

In the following, every point $p=(x,y)$ on the grid is replaced by a so-called \emph{oriented point}, which is a pair of points $P=\{t(p), b(p)\}$ with \emph{top} $t(p)=(x, y+\frac{\varepsilon}{2})$ and \emph{bottom} $b(p)=(x, y-\frac{\varepsilon}{2})$ with $\varepsilon:=10^{-6}$. 

Then, $\odil(P,P')$ is defined as the maximum dilation between one of the two points of $P$ and one of the points of $P'$.
We denote the minimum triangle containing one of $\{t(p), b(p)\}$ and one of $\{t(p'), b(p')\}$ by $\Delta^*(P,P')$. In this proof, the following bounds is used repeatedly to bound $\odil(P,P')$. It holds that
\begin{equation*}
    2|pp'|\leq |\Delta^*(P,P')| \leq 2|pp'|+2\varepsilon.
\end{equation*}


We add an edge between $t(p)$ and $b(p)$, its orientation encodes whether this points represent ``true'' or ``false''. W.l.o.g. we assume that an oriented edge from $b(p)$ to $t(p)$, thus an \emph{upwards} edge, represents ``true'' and a \emph{downwards} edge represents ``false''. When this is not the case, we can achieve this by flipping the orientation of all edges.

Edges in the plane embedding of our formula graph $G_\varphi$ is replaced by \emph{wire} gadgets. First, we add (a polynomial number of) grid points on the edge such that all edges have length $1$. Then, we create a wire as in \cref{fig:wire-positive}. Note that wires propagate the orientation of oriented points -- if two \emph{direct neighbours along a wire} have different orientations, their dilation would be $2-\varepsilon> 1.043$, since their shortest closed walk needs to go through an additional point. If they have the same orientation, their dilation is smaller than $1+2\varepsilon$. To switch a signal (for a negated variable in a clause), we start a wire as in \cref{fig:wire-negation}.

\begin{figure}[ht]
    \centering
\begin{minipage}{\linewidth}
    \centering
   \includegraphics{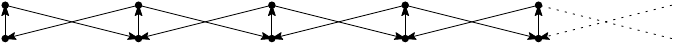}
    \caption{A wire where oriented points are oriented upwards}
    \label{fig:wire-positive}
\end{minipage}\hfill
 \begin{minipage}{\linewidth}
    \centering
    \includegraphics{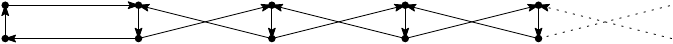}
    \caption{A wire where the orientation is switched to negate the signal}
    \label{fig:wire-negation}
\end{minipage}
\end{figure}
    
    To ensure that all clause gadgets encode the same orientation of oriented points as ``true'', we add a \emph{tree of knowledge}.
    This is a tree with vertices on a $1\times 1$-grid shifted by $(0.5, 0.5)$ relative to the grid of $G_\varphi$. Again, we use wires as edges. The tree has two leaves per clause (see \cref{fig:tree-of-knowledge}).  W.l.o.g. we assume that all oriented points of the tree of knowledge are oriented upwards (thus ``true''). 

    \begin{figure}[ht]
        \centering
        \includegraphics{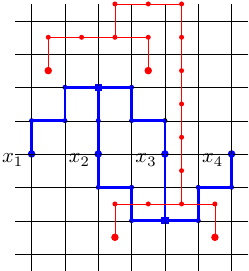}
        \caption{$G_\varphi$ (blue) and its underlying grid together with a tree of knowledge (red)}
        \label{fig:tree-of-knowledge}
    \end{figure}

   All oriented points, which are not direct neighbours along a wire, are linked by a $K_{2,2}$, that is by the four possible edges between top-top, bottom-bottom, top-bottom, bottom-top. This ensures $\odil(P,P')\leq 1+\varepsilon$ for not direct neighbours (compare to \cref{fig:k4-orientation}).

 \begin{figure}[ht]
    \centering
    \hfill\includegraphics[page=1]{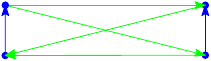} \hfill \includegraphics[page=2]{figures/k22-colored.pdf}\hfill~
         \caption{$K_{2,2}$ for oriented points with same orientation (left) and different orientation (right)}
    \label{fig:k4-orientation}
    \end{figure}
    
      \begin{figure}[ht]
        \centering
        \includegraphics{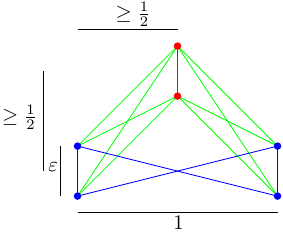}
        \caption{A wire (blue) can not be shortcut by $K_{2,2}$ gadgets  (green)}
        \label{fig:short-cut-wire-by-k4}
    \end{figure}    

    For an oriented point $P$, the shortest closed walk containing $t(p)$ and $b(p)$ is though the closest point $p'$ of $p$. Therefore, it holds that $\odil(t(p),b(p))\leq 1+2\varepsilon$, both if $p$ and $p'$ are direct neighbours (shortest closed walk is through a wire) and not (through a $K_{2,2}$). 

Adding these $K_{2,2}$ gadgets does not affect the functionality of wire gadgets: 
Let $p$ and $p'$ be direct neighbours and $p^*$ a third point.
Since $p^*$ has at least distance $(0.5, 0.5)$ to $p$ and $p'$, a walk from $p$ to $p'$ visiting $p^*$ has at least length $\sqrt{2}>1+2\varepsilon$. Therefore, a wire can not be shortcut by $K_{2,2}$ gadgets (compare to \cref{fig:short-cut-wire-by-k4}).


   For every clause, its two leaves in the tree of knowledge are not linked by a $K_{2,2,}$, but rather by a \emph{clause gadget}. \cref{fig:clause-overview} shows the two leaves of the tree at a clause, and $G_\varphi$ at the clause. We can assume that $G_\varphi$ is embedded as shown, in particular leaving the area directly above the clause empty.
    We show how such a clause gadget looks like in \cref{fig:clause-overview-ellipse}, more detailed in \cref{fig:detailed-clause-gadget}.

    The oriented points $L$ (left), $R$ (right) and $B$ (bottom) are the ends of the variable wires of a clause. They are placed such that they lie just inside an ellipse with the leaves $l_1$ and $l_2$ as foci, and without any other points in the ellipse (compare to \cref{fig:clause-overview,fig:clause-overview-ellipse}). 
    As proven later, the gadget as shown in \cref{fig:detailed-clause-gadget} guarantees that to obtain $\odil(L_1,L_2)\leq t'$ one of $\{L,R,B\}$ lies on a shortest closed walk containing $L_1$ and $L_2$ and the orientation of that point has to be the same as of $L_1$ and $L_2$  (thus, the literal is ``true'').
    
    For each of $\{L,R,B\}$ there exists a \emph{satellite} point, which is an oriented point on the variable wire, which is close but outside the ellipse. Its purpose is to ensure that the oriented dilation of $L_1$ (and likewise $L_2$) with $L$, $B$ and $R$ is below $t'$, even if the orientation of that point is different as of $L_1$ and $L_2$ (thus, the corresponding literal does not satisfy the clause). 
   As described before, a $K_{2,2}$ exists between all non-neighbouring points, except the leaves $L_1$ and $L_2$. 

 \begin{figure}[ht]
    \begin{minipage}[b]{.45\linewidth}
         \centering
        \includegraphics{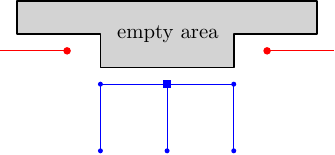}
        \caption{Embedded clause}
        \label{fig:clause-overview}
    \end{minipage}\hfill
    \begin{minipage}[b]{.45\linewidth}
         \centering
        \includegraphics{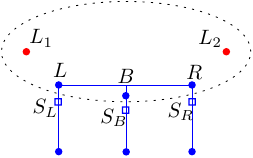}
        \caption{Clause gadget}
        \label{fig:clause-overview-ellipse}
    \end{minipage}
\end{figure}

    \begin{figure}[ht]
        \centering
        \includegraphics{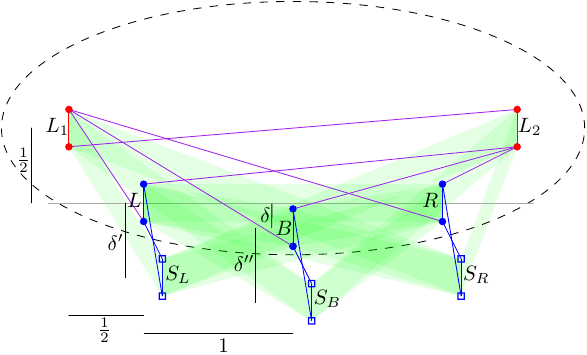}
        \caption{Detailed clause gadget. $K_{2,2}$ gadgets are indicated by green parallelograms.}
        \label{fig:detailed-clause-gadget}
    \end{figure}

  In the following, we set values for $\delta, \delta'$ and $\delta''$ (compare to \cref{fig:detailed-clause-gadget}) and thus place the points $L$, $B$ and $R$ exactly.

    First, $\delta'$ has to be set such that the satellite is not in the ellipse, but the dilation of the leaves and the left and right wire ends is bounded by $t'$, even if the shortest closed walk is through the $K_{2,2}$ between the leaf and the satellite. Thus, we set $\delta':=0.0152$ and achieve
    \begin{align*}
        \odil(L_1,L)&\leq \frac{2|l_1s_l|+2\delta'+8\varepsilon}{2|l_1l|}
        \leq \frac{|l_1l|+|ls_l|+\delta'+2\varepsilon}{|l_1l|}\\
        &\leq 1+ \frac{2\delta'+4\varepsilon}{\sqrt{\left(\frac{1}{2}\right)^2+\left(\frac{1}{2}\right)^2}}
         \leq 1+2\sqrt{2}\delta'+ 6\varepsilon 
        \leq 1.043.
    \end{align*}

    Analogously, we can upper bound $\odil(L_2,L)$, $\odil(L_1,R)$ and $\odil(L_2,R)$ by $1.043$.

       As can be seen in \cref{fig:clause-overview-ellipse,fig:detailed-clause-gadget}, the oriented point $B$ at the grid point $b$ has to be shifted by $(0,-\delta)$ from the original grid such that it is only just contained in the ellipse. Thus, $\delta$ has to be set such that $|l_1b|+ |bl_2| \leq |l_1l|+ |ll_2|$. This is true for $\delta:=0.1335$: 
       \begin{equation*}
           |l_1b|+ |bl_2|= 2\sqrt{\left(\frac{1}{2}+\delta\right)^2+\left(\frac{3}{2}\right)^2} \leq  |l_1l|+ |ll_2|= \sqrt{\left(\frac{1}{2}\right)^2+\left(\frac{1}{2}\right)^2}+\sqrt{\left(\frac{1}{2}\right)^2+\left(\frac{5}{2}\right)^2}.
       \end{equation*}

    Again, the wire from below to $B$ is disrupted by a  $\delta''$ close satellite which is not contained in the ellipse, to ensure that the dilation of the leaves and $B$ is bounded by $t'$, even if the shortest closed walk is through the $K_{2,2}$ to the satellite.
     Note that $\delta''$ is larger $\delta'$, as $B$ does not lay exactly on the ellipse. 
 Thus, we set $\delta'':= 0.0304$ and achieve
\begin{align*}
    \odil(L_1,B)&\leq \frac{2|l_1s_b|+2\delta''+8\varepsilon}{2|l_1b|}\leq  \frac{|l_1b|+|bs_b|+\delta''+4\varepsilon}{|l_1b|}= 1+\frac{2\delta''+4\varepsilon}{\sqrt{\left(\frac{1}{2}+\delta\right)^2+\left(\frac{3}{2}\right)^2}}\\
    &\leq 1+\frac{2\delta''+4\varepsilon}{\sqrt{2+\delta}}\leq 1+\frac{2\delta''+4\varepsilon}{\sqrt{2}}\leq 1+\sqrt{2}\delta''+3\varepsilon
    \leq 1.043
\end{align*}
 
 Since $\odil(L_1,B)=\odil(L_2,B)$, it also holds that $\odil(L_2,B)\leq 1.043$.

The clause gadget links $L_1$ and $L_2$ in a certain way: We add edges between the two leaves $L_1,L_2$ and wire end points $L,B,R$ such that for every incoming wire there is a closed walk through the two leaves and the wire end point (see \cref{fig:detailed-clause-gadget}). If $L$ has the same orientation as the tree of knowledge (i.e.\ the literal is true), the dilation is

\begin{equation*}
    \odil(L_1,L_2)\leq \frac{|l_1l|+|ll_2|+|l_1l_2|+6\varepsilon}{2|l_1l_2|}
    =\frac{{\sqrt{\left(\frac{1}{2}\right)^2+\left(\frac{1}{2}\right)^2}}+{\sqrt{\left(\frac{1}{2}\right)^2+\left(\frac{5}{2}\right)^2}}+3}{6}+\varepsilon
    \leq 1.043.
\end{equation*}

   The same holds if $R$ has the same orientation as the tree of knowledge.

       If $B$ has the same orientations as the tree of knowledge, the dilation is
       \begin{equation*}
           \odil(L_1,L_2)  \leq \frac{|l_1b|+|bl_2|+|l_1l_2|+6\varepsilon}{2|l_1l_2|}= \frac{2\sqrt{\left(\frac{1}{2}+\delta\right)^2+\left(\frac{3}{2}\right)^2}+3}{6}+\varepsilon
           \leq 1.043.
       \end{equation*}

     If $L$, $B$ and $R$ have a different orientation than the tree of knowledge (i.e.\ every literal is false), by triangle inequality there are the following possible shortest closed walks through $L_1$ and $L_2$ which lead to three lower bounds for $\odil(L_1,L_2)$:
    \begin{itemize}
        \item A closed walk is from $L_1$ to the satellite of $L$ (or $R$ symmetric) to $L_2$ and back to $L_1$:
        \begin{align*}
             \odil(L_1,L_2)&\geq  \frac{|l_1s_l|+|s_ll_2|+|l_1l_2|}{2 |l_1l_2|+2\varepsilon}\\
             &\geq \frac{\sqrt{(\frac{1}{2}+\delta')^2+\left(\frac{1}{2}\right)^2}+\sqrt{\left(\frac{1}{2}+\delta'\right)^2+\left(\frac{5}{2}\right)^2}+3}{6}\cdot (1-\varepsilon)
             > 1.043.
        \end{align*}

        \item A closed walk is from $L_1$ to the satellite of $B$ to $L_2$ and back to $L_1$:
        \begin{equation*}
            \odil(L_1,L_2)\geq  \frac{|l_1s_b|+|s_bl_2|+|l_1l_2|}{2 |l_1l_2|+2\varepsilon}
            \geq\frac{2\sqrt{\left(\frac{1}{2}+\delta+\delta''\right)^2+\left(\frac{3}{2}\right)^2}+3}{6}\cdot (1-\varepsilon)%
            > 1.043.
        \end{equation*}
        \item A closed walk is from $L_1$ to $L$ to $B$ to $L_2$ and back to $L_1$ (or symmetric through $B$ to $R$):
        \begin{align*}
            \odil(L_1,L_2)&\geq  \frac{|l_1l|+|lb|+|bl_2|+|l_1l_2|}{2 |l_1l_2|+2\varepsilon}\\
            &\geq \frac{\sqrt{\left(\frac{1}{2}\right)^2+\left(\frac{1}{2}\right)^2}+\sqrt{\delta^2+1}+\sqrt{\left(\frac{1}{2}+\delta\right)^2+\left(\frac{3}{2}\right)^2}+3}{6}\cdot (1-\varepsilon) 
            > 1.043.
        \end{align*}
    \end{itemize}

\noindent
Thus, we obtain the following conditions: 
    \begin{itemize}
    \item If one of the oriented points $L$, $B$, $R$ is oriented upwards, the dilation between $L_1$ and $L_2$ is smaller than $t' := 1.043$.
    \item If none of the oriented points $L$, $B$, $R$ is oriented upwards, the shortest closed walk containing $L_1$ and $L_2$ 
either leaves the ellipse or takes at least to points from $\{L,B,R\}$ and thus their dilation is greater than $t'$.   
\end{itemize}

Therefore, formula $\varphi$ is satisfiable if and only if there exists an orientation of our constructed graph with dilation at most $1.043$.
    \end{proof}

Following the construction in the proof of \cref{theo:np-hard-decision-orient}, \cref{fig:example-details} illustrates the graph for the formula $\varphi= (x_1\lor \neg x_2 \lor \neq x_3)\land (x_2\lor \neg x_3 \lor x_4)$ (see also \cref{fig:example-planar-3SAT,fig:tree-of-knowledge}).

\begin{figure}[ht]
    \centering
        \includegraphics{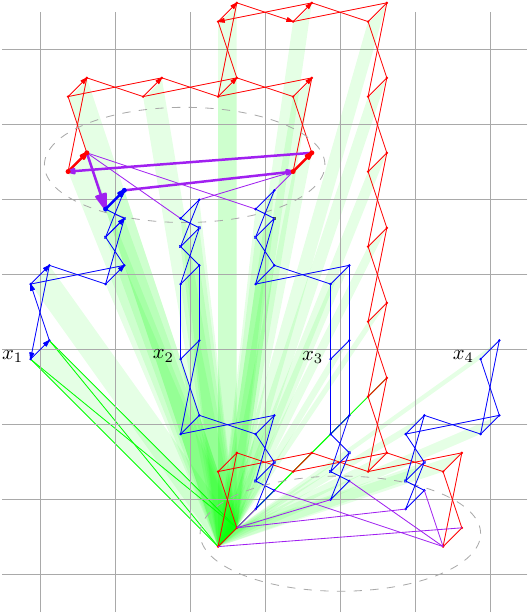}
    \caption{Graph constructed for $\varphi= (x_1\lor \neg x_2 \lor \neg x_3)\land (x_2\lor \neg x_3 \lor x_4)$. 
    For visibility, oriented points are placed diagonally instead of vertically. Only for one point, all $K_{2,2}$-gadgets are indicated by green parallelograms.
    If the oriented point at the end of the wire from $x_1$ (blue) is -- as indicated -- oriented the same way as the tree of knowledge (red), this corresponds to setting it to true, resulting in an oriented closed walk in the clause gadget (purple) that gives a dilation smaller than 1.043.
    }
    \label{fig:example-details}
\end{figure}

\section{Oriented Dilation}\label{sec:odil}

Computing the dilation of a given graph is related to the all-pairs shortest paths problem. There is a well-known cubic-time algorithm for APSP by Floyd and Warshall~\cite{Floyd62.APSP,Warshall62.APSP}.
By using their algorithm (and computing the minimum triangles naively), the oriented dilation of a given graph with $n$ points can be computed in $\bigo(n^3)$ time. 

We prove APSP-hardness for computing the oriented dilation. Since the APSP-problem is believed to need truly cubic time~\cite{Reddy16.APSP}, it is unlikely that computing the oriented dilation for a given graph is possible in subcubic time. 

Our hardness proof is a subcubic reduction to $\textsc{MinimumTriangle}$, which is the problem of finding a minimum weight cycle in an undirected graph with weights in $[7M, 9M]$, where the minimum weight cycle is a triangle. For two computational problems $A$ and $B$, a \emph{subcubic reduction} is a reduction from $A$ to $B$ where any subcubic time algorithm for $A$ would imply a subcubic time algorithm for $B$ (for a more thorough definition see~\cite{WilliamsW18}). 
 We use the notation $A\leq_3 B$ to denote the existence of a subcubic reduction from $A$ to $B$. 
 
Williams and Williams~\cite{WilliamsW18} prove the APSP-hardness for the problem of finding a minimum weight cycle in an undirected weighted graph. Their proof also implies the following hardness result
(see \cref{sec:williamsproof} for details).

\begin{restatable}[$\textsc{APSP} \leq_3 \textsc{MinimumTriangle}$]{theorem}{minimumTriangle}
\label{theo:minimumtriangle-APSP-hard}
Let $G$ be an undirected graph with weights in $[7M, 9M]$ whose minimum weight cycle is a triangle.
Then, it is APSP-hard to find a minimum weight triangle in $G$. 
\end{restatable}

By reduction to $\textsc{MinimumTriangle}$, we prove APSP-hardness for $\textsc{ODIL}$, which is the problem of computing the oriented dilation of a given oriented graph in a metric space.

\begin{theorem}[$\textsc{APSP} \leq_3 \textsc{Odil}$] 
Let $\oriented{G}$ be an oriented metric graph. It is APSP-hard to compute the oriented dilation of $\oriented{G}$. 
\end{theorem}
\begin{proof}
    Due to \cref{theo:minimumtriangle-APSP-hard}, it is sufficient to prove $\textsc{MinimumTriangle} \leq_3 \textsc{Odil}$.
    
    Let $G=(P,E)$ be an undirected graph with $|P|=n$ and weights $w: E \rightarrow [7M, 9M]$ in which the minimum weight cycle $C$ of $G$ is a triangle. Note that the weights of $G$ fulfil the triangle inequality.

We define a metric space $(P',|\cdot|)$ with $P'=P \cup \{u,v\}$ by the following distances:
  \begin{itemize}
         \item $|pq|=w(p,q)+15M$, for $\{p,q\}\in E$,
         \item $|pq|=14M+15M=29M$, for $\{p,q\}\notin E$, and
        \item $|up|=|vp|=|vu|=15M+15M=30M$ for $p\in P$.
    \end{itemize}   

  So, for each missing edge of a complete graph on $P$, we assign its weight to be $14M$. Therefore, we obtain a metric with distances in $\left[7M,14M\right]$.
  Then, we add the points $u$ and $v$ with distance $15M$ to each point in $p\in P$. 
 To  preserve the metric, we increase the distance of every tuple by $15M$.
 This construction is visualised in \cref{fig:APSP-hardness-sketch}.
 
 \begin{figure}[ht]
    \centering
    \includegraphics{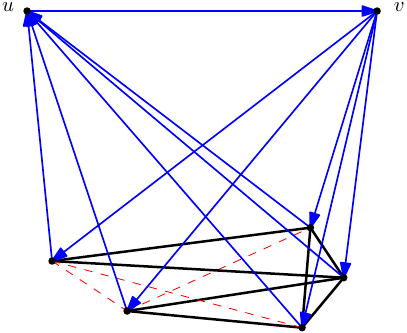}
    \caption{Graph construction in the proof of \cref{theo:minimumtriangle-APSP-hard}. Given an undirected graph on $P$~(black), we add two points $u,v$ with a large distance to all other points and augment it to a complete graph (red, dashed). The oriented graph $\oriented{G}$ (blue) ensures that every shortest closed walk containing $p$ and $q$ in $P$ visits $u$ and $v$.}
    \label{fig:APSP-hardness-sketch}
\end{figure}
 
    Let $\oriented{G}=(P',\oriented{E})$ be an oriented graph with 
     $\oriented{E}=\{(p,u),(v,p)\mid p\in P\} \cup \{(u,v)\}.$

We prove that the weight $w(C)$ of the minimum triangle can be computed by computing $\odil(\oriented{G})$. In particular, we show that $w(C)=\frac{180M}{\odil(\oriented{G})}-45M$.
    
    We proceed by computing the oriented dilation of $\oriented{G}$. 
    We first note that all edges of $\oriented{G}$ have weight $30M$. 
    
    For every tuple of $u$ or $v$ and a point $p\in P$, the shortest closed walk is a triangle of length  
   $|C_{\oriented{G}}(u,p)|=|C_{\oriented{G}}(v,p)|=|C_{\oriented{G}}(u,v)|=3\cdot 30M=90M$. 
   Their minimum triangle contains two edges of length $30M$.  
    Therefore, their dilation is
    \begin{equation*}
      \odil(u,p)= \frac{|C_{\oriented{G}}(u,p)|}{|\Delta^*(u,p)|}\leq \frac{90M}{2 \cdot 30M}=1.5.
    \end{equation*}
    Analogously, it holds that $\odil(v,p)\leq 1.5$ for $p\in P$.
   
    For $u$ and $v$, it holds that $\odil(u,v)=1$, because every triangle in $P'$ containing $u$ and $v$ is an oriented closed walk in $\oriented{G}$.

For two distinct points $p,q\in P$, the shortest closed walk containing $p$ and $q$ visits $u$ and $v$ and has length $|C_{\oriented{G}}(p,q)|=6\cdot 30M=180M$. For their minimum triangle holds 
       $ |\Delta^*(p,q)|\leq 3\cdot29M=87M$. 
    Since their dilation is larger than $2$,
    a worst-case pair of $\oriented{G}$ must be a pair in~$P$. 
    Therefore, the oriented dilation of the constructed graph $\oriented{G}$ is
    \begin{equation*}
        \odil(\oriented{G})=\underset{p,q\in P}{\max} \frac{|C_{\oriented{G}}(p,q)|}{|\Delta^*(p,q)|} 
        =\frac{180M}{\underset{p,q,p^*\in P}{\min} w'(p,q)+w'(q,p^*)+w'(p,p^*)}  >2.
    \end{equation*}
     Since $C$ is a triangle and $w(C)\leq 27M$, it holds that
    \begin{equation*}
        \underset{p,q,p^*\in P}{\min}|pq|+|qp^*|+|p,p^*| = \underset{p,q,p^*\in P}{\min} w(p,q)+w(q,p^*)+w(p,p^*)+3\cdot 15M= w(C)+45M.
    \end{equation*}
    Therefore, the oriented dilation of $\oriented{G}$ is achieved by the minimum-length triangle $C$.
    Consequently, given the oriented dilation of $\oriented{G}$, the length of $C$ is
       $w(C)=\frac{180M}{\odil(\oriented{G})}-45M$.

The most expensive operation in our reduction is defining the metric on $\bigo(n^2)$ points with distances in $\left[22M,30M\right]$. Following the definitions in~\cite{WilliamsW18}, all instances have size $\bigo(n^2\operatorname{polylog}(M))$ and all operations take subcubic time.
  
    Thus, computing the oriented dilation of an oriented metric graph is at least as hard as  finding a minimum weight triangle in an undirected graph.
    \end{proof}

To compute the oriented dilation, testing all $\binom{n}{2}$ point tuples seems to be unavoidable. However, to approximate the oriented dilation,  we show that, 
a linear set of tuples suffices.

\begin{theorem}\label{theo:approx-odil}
Let $P$ be a set of $n$ points in $\mathbb{R}^d$, let $\oriented{G}$ be an oriented graph on $P$, and let $\varepsilon>0$ be a real number.
Assume that, in $T(n)$ time, we can construct a data structure such that,
for any two query points $p$ and $q$ in $P$, we can return, in $f(n)$ time, a $k$-approximation to the length of a shortest path in $G$ between $p$ and $q$. Then we can compute a value $\overline{\odil}$ in $\bigo(T(n) + n(\log n +f(n)))$ time, such that
\begin{equation*}
    (1-\varepsilon)\cdot \odil(\oriented{G}) \leq  \overline{\odil}  \leq k\cdot \odil(\oriented{G}).
\end{equation*}
\end{theorem}
\begin{proof}
Let the constants in \cref{alg:approx-odil} be equal to $\varepsilon_1= {\varepsilon}/{2}$ and $s= {28}/{\varepsilon}$.
Let $t^*$ be the dilation of a given oriented graph $\oriented{G}$ on $P$.
We prove that 
 \[ (1-\varepsilon)t^* \leq  \overline{\odil} \leq kt^*,\]
where $\overline{\odil}$ is the dilation approximated by \cref{alg:approx-odil}. 

Throughout this proof, let $p,q\in P$ be a worst-case pair in $\oriented{G}$, i.e., $\odil(p,q)=t^*$.

\subparagraph{Upper bound.}
    For each well-separated pair $\{A,B\}$, \cref{alg:approx-odil} picks two points $a\in A$ and $b\in B$, approximates the length of a shortest closed walk $C_{\oriented{G}}(a,b)$ containing $a$ and $b$ and their minimum triangle $\Delta^*(a,b)$. 
    It holds that
    \begin{equation*}
        \overline{\odil}
        \leq \underset{a\in A, b\in B}{\max} \frac{k\cdot |C_{\oriented{G}}(a,b)|}{|\Delta^*(a,b)|} = k\cdot \underset{a\in A, b\in B}{\max} \odil(a,b) \leq k \cdot\odil(\oriented{G}) = k \cdot t^*.
    \end{equation*}
    
  \subparagraph{Lower bound.} 
   Let $\{A,B\}$ be the well-separated pair with $p\in A$ and $q\in B$. We distinguish cases based on whether $p$ and/or $q$ are picked by \cref{alg:approx-odil} while processing the pair $\{A,B\}$ or not.

\subparagraph{Case 1:} neither $p$ nor $q$ is picked by \cref{alg:approx-odil}.

   Let $a\in A$ and $b\in B$ be the points picked by \cref{alg:approx-odil}. 

First, we bound the length of a shortest closed walk containing $p$ and $q$ by
\begin{equation*}
    |C_{\oriented{G}}(p,q)|\leq |C_{\oriented{G}}(p,a)|+|C_{\oriented{G}}(a,b)|+|C_{\oriented{G}}(b,q)|\leq t^*|\Delta^*(p,a)|+|C_{\oriented{G}}(a,b)|+t^*|\Delta^*(b,q)|.
\end{equation*}

Since $p$ and $q$ are both not picked, this implies $|A|\geq 3$ and $|B|\geq 3$. Analogously to \cref{eq:bound-Delta-same-set-pair}, we bound $|\Delta^*(p,a)|$ and $|\Delta^*(b,q)|$ both by $3/s |\Delta^*(p,q)|$. It follows that
    \begin{equation*}
        |C_{\oriented{G}}(p,q)|\leq 2\cdot  3/s\cdot t^* \cdot|\Delta^*(p,q)|+|C_{\oriented{G}}(a,b)| = 6/s \cdot |C_{\oriented{G}}(p,q)|+|C_{\oriented{G}}(a,b)|  .
    \end{equation*}

  Due to \cref{theo:delta-ab-WSPD}, we use $|\Delta^*(a,b)|\leq (1+8/s) |\Delta^*(p,q)|$ to bound
\begin{equation*}
 |C_{\oriented{G}}(p,q)|\leq \frac{s}{s-6}  |C_{\oriented{G}}(a,b)|= \frac{s}{s-6}  \odil(a,b) |\Delta^*(a,b)| \leq  \frac{s+8}{s-6}  \odil(a,b) |\Delta^*(p,q)|.
\end{equation*}

Since $t^*=\frac{|C_{\oriented{G}}(p,q)|}{|\Delta^*(p,q)|}$, this leads us to a lower bound of the dilation of $a$ and $b$, which is
\begin{equation*}
     \odil(a,b) \geq  \frac{s-6}{s+8} \cdot t^* \geq \left(1- 14/s\right) t^* .
\end{equation*}

For the picked points $a$ and $b$, \cref{alg:approx-odil} computes an approximation of $C_{\oriented{G}}(a,b)$ and a $(1+\varepsilon_1)$-approximation of $\Delta^*(p,b)$. It holds that
 \begin{equation*}
        \overline{\odil} \geq \frac{|C_{\oriented{G}}(a,b)|}{(1+\varepsilon_1)|\Delta^*(a,b)|} \geq (1-\varepsilon_1) \odil(a,b) \geq (1-\varepsilon_1)\left(1- 14/s\right) t^*.
    \end{equation*}

   Since  $\varepsilon_1=\varepsilon/2$ and $s=28/\varepsilon$, we conclude that 
 \begin{equation*}
        \overline{\odil} \geq 
        \left(1-\varepsilon/2\right)^2\cdot t^*
       \geq (1-\varepsilon)\cdot t^*.
    \end{equation*}

  \subparagraph{Case 2:} either $p$ or $q$ is picked by \cref{alg:approx-odil}.

Without loss of generality, let $p$ be picked for $A$ and a distinct point $b$ be picked for $B$.

Analogously to Case 1, we can show 
\begin{equation*}
    |C_{\oriented{G}}(p,q)|\leq |C_{\oriented{G}}(p,b)|+|C_{\oriented{G}}(b,q)|\leq  \frac{s+8}{s-6}  \odil(a,b) |\Delta^*(p,q)|, 
\end{equation*}
which leads us to 
 \begin{equation*}
        \overline{\odil} \geq (1-\varepsilon_1) \odil(a,b) \geq (1-\varepsilon) t^*.
    \end{equation*}

 \subparagraph{Case 3:} both $p$ and $q$ are picked by \cref{alg:approx-odil}.
  
Since \cref{alg:approx-odil} computes an approximation of $C_{\oriented{G}}(p,q)$ and $(1+\varepsilon_1)$-approximation of $\Delta^*(p,q)$, it holds that
 \begin{equation*}
        \overline{\odil} \geq \frac{|C_{\oriented{G}}(p,q)|}{(1+\varepsilon_1)|\Delta^*(p,q)|}\geq (1-\varepsilon_1)  \odil(p,q)  \geq (1-\varepsilon)t^*,
    \end{equation*}
where the last inequality  follows from  $\varepsilon_1=\varepsilon/2$.\bigskip

   The preprocessing is dominated by computing the data structure for answering approximate shortest pair and    
   a well-separated pair decomposition in $\bigo(T(n) + n \log n)$ time (\cref{theo:wspd-construction}).
     For each of the $\bigo(n)$ pairs, we pick two points $a$ and $b$ and approximate their smallest triangle in $\bigo(\log n)$ time (compare to~\cref{theo:approx-triangle-queries}). The runtime to approximate the length of a shortest closed walk containing $a$ and $b$ depends on the runtime $f(n)$ of the used approximation algorithm. Therefore, the oriented dilation of a given graph can be approximated in $\bigo(T(n) + n(\log n +f(n)))$ time.
\end{proof}

Our approximation algorithm (\cref{alg:approx-odil}) uses the WSPD to achieve a suitable selection of $\bigo(n)$ point tuple that needs to be considered. 
The algorithm and its analysis work similar to our $(2+\varepsilon)$-spanner construction (compare to \cref{alg:wspd-oriented}). 
 For each selected tuple, the algorithm approximates their oriented dilation by approximating their minimum triangle and shortest closed walk.

The precision and runtime of this approximation depend on approximating the length of a shortest path between two query points. 
The point-to-point shortest path problem in directed graphs is an extensively studied problem (see \cite{goldberg05.P2P} for a survey). 
Alternatively, shortest path queries can be preprocessed via approximate APSP (see~\cite{DoryFKNWV24.APSP,SahaY24.APSP,Zwick02.APSP}).
For planar directed graphs, the authors in~\cite{ArikatiCCDSZ96} present a data structure build in $\bigo(n^2 / r)$ time, such that the exact length of a shortest path between two query points can be returned in $\bigo(\sqrt{r})$ time. 
Setting $r=n$ minimizes the runtime of our \cref{alg:approx-odil} to $\bigo(n \sqrt{n})$. 
It should be noted that even with an exact shortest path queries ($k=1$) the oriented dilation remains an approximation due to the minimum-perimeter triangle approximation.


\begin{algorithm}[ht]
	\caption{Approximate Oriented Dilation}\label{alg:approx-odil}
     \begin{algorithmic}
    		\REQUIRE Oriented graph $\oriented{G}=(P,\oriented{E})$, positive reals $\varepsilon$, $\varepsilon_1$, $s$, parameter $k$
    		\ENSURE Value $\overline{\odil}$ , such that $(1-\varepsilon)\odil(\oriented{G}) \leq  \overline{\odil} \leq k \odil(\oriented{G})$
            \STATE Compute $s$-WSPD with the pairs $\{A_i,B_i\}$ for $1 \leq i \leq m$
            \STATE Compute the data structure $M$ on $P$ to approximate nearest neighbours (see \cref{theo:nearest-neighbour-query}~\cite{AryaMountDS1998})%
            \STATE $\overline{\odil}=1$
    
             \FOR{$i=1$ \textbf{to} $m$}
            \STATE $A=$ Pick $\min\{|A_i|,2\}$ points of $A_i$
            \STATE $B=$ Pick $\min\{|B_i|,2\}$ points of $B_i$
            \FORALL{$\{a,b\}\in A\times B$}
                \STATE $\overline{d}(a,b)=$ $k$-approximation of the length of a shortest path from $a$ to $b$ in $\oriented{G}$
                 \STATE $\overline{d}(b,a)=$ $k$-approximation of the length of a shortest path from $b$ to $a$ in $\oriented{G}$
                \STATE $\Delta_{abc}=$ $(1+\varepsilon_1)$-approximation of $\Delta^*(a,b)$ by \cref{alg:approx-min-triangle} given $M$, $a$, $b$, $\varepsilon_1$
                 \STATE $\overline{\odil}=\max\bigg\{ \frac{\overline{d}(a,b)+\overline{d}(b,a)}{|\Delta_{abc}|}, \overline{\odil}\bigg\}$
                 \ENDFOR
        \ENDFOR
        \RETURN  $\overline{\odil}$
\end{algorithmic}
\end{algorithm}

\section{Conclusion and Outlook}
We present an algorithm for constructing a sparse oriented $(2+\varepsilon)$-spanners for multidimensional point sets. 
Our algorithm computes such a spanner efficiently in $O(n \log n)$ time for point sets of size $n$. 
In contrast,~\cite{ESA23} presents a plane oriented $\mathcal{O}(1)$-spanner, but only for points in convex position. Developing algorithms for constructing plane oriented $\bigo(1)$-spanners for general two-dimensional point sets remains an open problem. Another natural open problem is to improve upon $t = 2+\varepsilon$. For this, the question arises: Does every point set admit an oriented $t$-spanner with $t<2$. Even for complete oriented graphs this is open.

Both algorithms could be improved by optimally orienting the underlying undirected graph, which is constructed first.
We discard this idea by proving that, given an undirected graph, it is NP-hard to find its minimum dilation orientation. 
In particular, is the problem NP-hard for plane graphs or for complete graphs?

In the second part of this paper, we study the problem of computing the oriented dilation of a given graph. We prove APSP-hardness of this problem for metric graphs. We complement this by a subcubic approximation algorithm. This algorithm depends on approximating the length of a shortest path between two given points. Therefore, improving the approximation of a shortest path between two given points in directed/oriented graphs is an interesting problem. The APSP-hardness of computing the oriented dilation of metric graphs seems to be dominated by the computation of a minimum-perimeter triangle. However, for Euclidean graphs, the minimum-perimeter triangle containing two given points can be computed in $o(n)$ time. This raises the question: is computing the oriented dilation of a Euclidean graph easier than for general metric graphs? 




Finally, as noted in~\cite{ESA23}, in many applications some bi-directed edges might be allowed. This opens up a whole new set of questions on the trade-off between dilation and the number of bi-directed edges. Since this is a generalisation of the oriented case, both of our hardness results also apply to such models.

\bibliography{bibliography-arxiv}

\appendix

\section{Subcubic Equivalences Between Path and Triangle Problems}\label{sec:williamsproof}

In \cite{WilliamsW18}, Vassilevska Williams and Williams proved for a list of various problems that either all of them have subcubic algorithms, or none of them do. This list includes
\begin{itemize}
    \item computing a shortest path for each pair of points on weighted graphs (known as APSP),
    \item finding a minimum weight cycle in an undirected graph with non-negative weights (called \textsc{MinimumCycle}), and
    \item detecting if a  weighted graph has a triangle of negative edge weight (called \textsc{NegativeTriangle}).
\end{itemize}

\begin{theorem}[\textsc{NegativeTriangle} $\leq_3$ \textsc{MinimumCycle}~\cite{WilliamsW18}]
Let $G$ an undirected graph with weights in $\left[1,M\right]$.
Finding a minimum weight cycle in an undirected graph with non-negative weights is at least as hard as deciding if a weighted graph contains a triangle of negative weight. 
\end{theorem}
\begin{proof}
    The following proof is restated from~\cite{WilliamsW18}. 
    Let $G=(V,E)$ be  a given undirected graph with weights $w: E\rightarrow \left[-M,M\right]$. We define an undirected graph $G'=(V,E)$, which is just $G$, but with weights $w': E\rightarrow \left[7M,9M\right]$ defined as $w'(e)=w(e)+8M$. 
    Any cycle $C$ in $G$ with $k$ edges, has length $w'(C)=8Mk+w(C)$, and it holds that $7Mk\leq w'(C)\leq 9Mk$. 
    Therefore, every cycle $C$ with at least four edges has length $w'(C)\geq 28M$ and any triangle $\Delta$ has length $w'(\Delta)\leq 27M<28M$. 
    Thus, a minimum weight cycle in $G'$ is a minimum weight triangle in $G$.
\end{proof}

However, since the minimum weight cycle in the constructed graph $G'$ is a triangle, we immediately obtain the following theorem. 

\minimumTriangle*

\end{document}